\documentclass[a4paper,11pt]{article}

\usepackage{amsmath,amssymb,mathrsfs}
\usepackage[dvips]{graphicx}
\usepackage{amsthm,bm,
color}
\usepackage{graphicx}
\usepackage{wrapfig}

\makeatletter
\def\figcaption{\def\@captype{figure}\caption}
\makeatother

\makeatletter
 
  \@addtoreset{equation}{section}
 \makeatother

\newtheorem{theorem}{\bf Theorem}[section]
\newtheorem{lemma}[theorem]{\bf Lemma}
\newtheorem{proposition}[theorem]{\bf Proposition}
\newtheorem{corollary}[theorem]{\bf Corollary}

\newtheorem{example}{\bf Example}[section]
\newtheorem{remark}{\bf Remark}[section]
\newtheorem{definition}{\bf Definition}[section]

\title{Supersymmetric quantum walks \\ with chiral symmetry}
\author{
Akito Suzuki
\thanks{Division of Mathematics and Physics, 
Faculty of Engineering, Shinshu University, Wakasato, Nagano 380-8553, Japan,
e-mail: akito@shinshu-u.ac.jp
	}
}
\begin{document}

\maketitle

\begin{abstract}
Quantum walks have attracted attention as
a promising platform realizing topological phenomena
and many physicists have introduced various types of 
indices to characterize topologically protected bound states
that are robust against perturbations.
In this paper, we introduce an index from a supersymmetric point of view.
This allows us to define indices for all chiral symmetric quantum walks
such as multi-dimensional split-step quantum walks and quantum walks 
on graphs, for which there has been no index theory. 
Moreover, the index gives a lower bound on the number of bound states
robust against compact perturbations. 
We also calculate the index for several concrete examples 
including the unitary transformation that appears in Grover's search algorithm. 
\end{abstract}
\section{Introduction}
Quantum walks have attracted attention 
as sources of ideas for quantum algorithms 
\cite{Am0, Am1, AKR, MNRoS07, P13, SKW}. 
Motivated by Grover's quantum search algorithm \cite{Gr}, 
Szegedy \cite{Sz} quantized a Markov chain on a finite bipartite graph
and defined a quantum walk, 
which has been updated 
\cite{KPSS18, MNRoS07,MNRS09, HKSS14, HSS, Se13} 
to define quantum walks on general (possibly infinite) graphs. 
What is common to such quantum walks is 
to have an evolution operator  
defined as a product of two unitary involutions. 
\subsection{Spectral mapping and supersymmetry}
For two given unitary involutions $\varGamma$ and $C$ 
on a Hilbert space $\mathcal{H}$,
we can introduce a coisometry $d$ from $\mathcal{H}$ to 
another Hilbert space $\mathcal{K}$ 
and a self-adjoint operator $T = d \varGamma d^*$ 
so that $C = 2d^*d-1$ and $\|T\| \leq 1$. 
A fascinating property of the product $U = \varGamma C$ 
of the two unitary involutions is as follows. 
Let $\varphi: S^1 \to [-1,1]$ be defined as
$\varphi(z) = (z + z^{-1})/2$.  
Then the spectrum of $U$ and 
the preimage of the spectrum of $T$ under $\varphi$
coincide except for the points $+1$ and $-1$, {\it i.e.},
\begin{equation}
\label{specU} 
\sigma\left(U_0\right) 
	= \varphi^{-1}\left(\sigma\left(T_0\right)\right),  
\end{equation}
where $U_0= U|_{\ker (U^2-1)^\perp}$ and 
$T_0= T|_{\ker(T^2-1)^\perp}$ are the restrictions 
onto $\ker (U^2-1)^\perp$ and $\ker (T^2-1)^\perp$. 
This property is called the spectral mapping theorem 
for the product of two unitary involutions. 
As depicted in Fig. \ref{fig01},
$\sigma(U_0)$ 
is divided into two parts, {\it i.e.}, 
\[ \sigma(U_0) = g_+(\sigma(T_0)) \cup  g_-(\sigma(T_0)), \]
where $g_\pm (\xi) = e^{\pm i\arccos \xi}$
for $\xi \in [-1,1]$. 
Moreover,  $U_0$ is unitarily equivalent to
\[ e^{i \arccos T_0} \oplus e^{-i \arccos T_0} \]
(see \cite{FFSd, SS, HSS} for more details). 
This is a sign of supersymmetry. 
In this paper, 
we explore the supersymmetry of the two unitary involutions.  
 
\begin{wrapfigure}{r}{75mm}
\unitlength 0.1in
\begin{picture}( 25.9600, 22.9000)(  2.6000,-23.6000)
%
\special{pn 4}%
\special{ar 1668 1266 894 898  0.0000000 6.2831853}%
%
\special{pn 8}%
\special{pa 346 1266}%
\special{pa 2856 1266}%
\special{fp}%
\special{sh 1}%
\special{pa 2856 1266}%
\special{pa 2790 1246}%
\special{pa 2804 1266}%
\special{pa 2790 1286}%
\special{pa 2856 1266}%
\special{fp}%
%
\special{pn 8}%
\special{pa 1668 2360}%
\special{pa 1668 70}%
\special{fp}%
\special{sh 1}%
\special{pa 1668 70}%
\special{pa 1648 138}%
\special{pa 1668 124}%
\special{pa 1688 138}%
\special{pa 1668 70}%
\special{fp}%
%
\special{pn 20}%
\special{ar 1668 1266 894 898  4.3138110 5.8156142}%
%
\special{pn 20}%
\special{ar 1668 1274 894 898  0.4675711 1.9702648}%
%
\special{pn 8}%
\special{pa 1302 448}%
\special{pa 1302 2100}%
\special{dt 0.045}%
%
\special{pn 8}%
\special{pa 2472 878}%
\special{pa 2472 1660}%
\special{dt 0.045}%
%
\special{pn 20}%
\special{sh 1}%
\special{ar 936 1778 10 10 0  6.28318530717959E+0000}%
%
\special{pn 8}%
\special{pa 936 762}%
\special{pa 936 1778}%
\special{dt 0.045}%
\put(26.4100,-13.4500){\makebox(0,0)[lt]{1}}%
\put(5.5000,-13.4500){\makebox(0,0)[lt]{$-1$}}%
\put(18.9500,-2.5800){\makebox(0,0){{\footnotesize $\sigma(U_0)$}}}%
%
\special{pn 8}%
\special{pa 2472 1346}%
\special{pa 2472 1220}%
\special{fp}%
%
\special{pn 8}%
\special{pa 1302 1346}%
\special{pa 1302 1220}%
\special{fp}%
\put(2.6000,-4.0300){\makebox(0,0)[lt]{{\footnotesize Eigenvalue of $U_0$}}}%
\put(2.6000,-19.9200){\makebox(0,0)[lt]{{\footnotesize Eigenvalue of $U_0$}}}%
%
\special{pn 4}%
\special{pa 868 556}%
\special{pa 922 698}%
\special{fp}%
\special{sh 1}%
\special{pa 922 698}%
\special{pa 918 630}%
\special{pa 904 648}%
\special{pa 880 644}%
\special{pa 922 698}%
\special{fp}%
%
\special{pn 4}%
\special{pa 832 1948}%
\special{pa 914 1822}%
\special{fp}%
\special{sh 1}%
\special{pa 914 1822}%
\special{pa 860 1866}%
\special{pa 884 1866}%
\special{pa 894 1888}%
\special{pa 914 1822}%
\special{fp}%
%
\special{pn 13}%
\special{pa 1306 1274}%
\special{pa 2476 1274}%
\special{fp}%
%
\special{pn 20}%
\special{sh 1}%
\special{ar 940 1266 10 10 0  6.28318530717959E+0000}%
\special{sh 1}%
\special{ar 940 1266 10 10 0  6.28318530717959E+0000}%
\put(18.9800,-11.4300){\makebox(0,0){{\footnotesize $\sigma(T_0)$}}}%
%
\special{pn 8}%
\special{pa 1896 852}%
\special{pa 1896 574}%
\special{fp}%
\special{sh 1}%
\special{pa 1896 574}%
\special{pa 1876 640}%
\special{pa 1896 626}%
\special{pa 1916 640}%
\special{pa 1896 574}%
\special{fp}%
%
\special{pn 8}%
\special{pa 1896 1660}%
\special{pa 1896 1938}%
\special{fp}%
\special{sh 1}%
\special{pa 1896 1938}%
\special{pa 1916 1872}%
\special{pa 1896 1886}%
\special{pa 1876 1872}%
\special{pa 1896 1938}%
\special{fp}%
\put(19.4900,-7.1600){\makebox(0,0)[lb]{{\footnotesize $g_+$}}}%
\put(19.5800,-17.8600){\makebox(0,0)[lb]{{\footnotesize $g_-$}}}%
%
\special{pn 20}%
\special{sh 1}%
\special{ar 940 752 10 10 0  6.28318530717959E+0000}%
\end{picture}%
  \caption{Location of the spectrum $\sigma(U_0)$. 
  The functions 
  $g_\pm$ 
  map
  $\sigma(T_0) \subset [-1,1]$  
  onto $\sigma(U_0) \subset S^1$.
  }
  \label{fig01}
\end{wrapfigure}
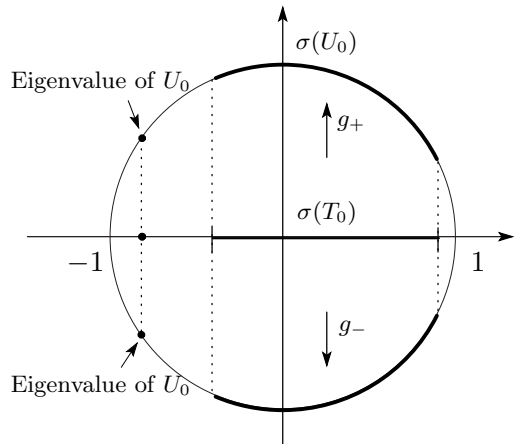

On the other hand, the quantum walks have also been viewed as 
promising platforms to realize topological phenomena \cite{Kit}. 
Kitagawa et al \cite{KBFRBKAD, KRBD} showed that 
one- and two-dimensional quantum walks exhibit topological phases
and experimentally realized topologically protected bound states. 
To this end, they employed a split-step quantum walk,
which possesses chiral symmetry, {\it i.e.},
the evolution operator $U$ satisfies
\begin{equation}
\label{chisym}
\varGamma U \varGamma = U^{-1} 
\end{equation}
with some unitary involution $\varGamma$. 
Asb\'{o}th and Obuse \cite{AsOb13} also studied 
the topological nature of a one-dimensional quantum walk
in  a chiral symmetric time frame (see also \cite{OA, OK}).
In the above studies, 
several types of topological indices were introduced  
in terms of winding numbers and Chern numbers
and they were used for characterizing the topological phenomena. 
Gross et al \cite{Gross} also established 
another index theory 
in terms of the flow of a walk \cite{Kitaev} 
and Cedzich et al \cite{Ced1, Ced2} studied topological classifications 
with various types of symmetry.
Topological phenomena for nonunitary PT-symmetric quantum walks 
were considered in \cite{MKO, XZB} 
and topological phenomena for periodically driven systems
were studied in \cite{AsTa14, KBRD, SaSch17}. 
Their definitions and proofs however  
deeply rely on the spatial dimension and geometry of the quantum walk. 
There has been no index theory that covers quantum walks on graphs 
and quantum walks for quantum algorithms.  
In the present paper, we establish index theory 
that can cover not only one and two-dimensional quantum walks 
but also such quantum walks. 
To this end, we first prove that
the evolution operator of every chiral symmetric quantum walk 
can be written as  a product of two unitary involutions
and it possesses supersymmetry. 
Then we define an index for such an evolution operator 
so that it coincides with the Witten index \cite{Witten}. 
\subsection{Index formula} 
To make it more precise, 
let $U$ obey chiral symmetry \eqref{chisym}. 
Then $C:=\varGamma U$ is a unitary involution
and hence $U = \varGamma C$ can be written as
a product of two unitary involutions. 
Actually, we can prove that
every quantum walk with an evolution operator 
represented as  a product of two unitary involutions
possesses chiral symmetry.   
Thus, we find that 
the above mentioned spectral mapping theorem
is applicable for any chiral symmetric quantum walk. 
Moreover, the spectral mapping theorem \cite{HSS} implies that
\[ \dim \ker(U \mp 1) = m_\pm + M_\pm, \]
where  $m_\pm = \dim \ker (T \mp 1)$
and $M_\pm = \dim \mathcal{B}_\pm$, 
and $\mathcal{B}_\pm := \ker (\varGamma \pm 1) \cap \ker d$ 
is called the birth eigenspaces \cite{HKSS14, MOS, Se13}. 
The supersymmetric structure is introduced as follows. 
From \eqref{chisym},
we observe  that
\begin{equation}
\label{def_Q} 
Q := \frac{1}{2i} [ \varGamma, C] 
\end{equation}
plays a role of supercharge:
$Q$ anticommutes with $\varGamma$,
{\it i.e.},  $\{\varGamma, Q\} = 0$.
Here $[ A, B] := AB- BA$ and  $\{A, B\} := AB+BA$
are the commutator and anticommutator of $A$ and $B$.   
From a standard argument of supersymmetric quantum mechanics,
the supersymmetric Hamiltonian $H := Q^2$ is decomposed into 
$H = H_+ \oplus H_-$
on $\ker(\varGamma-1) \oplus \ker(\varGamma+1)$. 
Then we define an index ${\rm ind}_\varGamma(U)$ for $U$ so that
${\rm ind}_\varGamma(U)$ agrees with the Witten index of $H$:
$\dim \ker H_+ - \dim \ker H_-$. 
The main result of this paper is the following index formula. 
\begin{equation}
\label{eq013}
{\rm ind}_\varGamma(U) 
 	= (M_- - m_-) - (M_+ - m_+).  
\end{equation}
It is clear from \eqref{eq013} that the absolute value on
${\rm ind}_\Gamma(U)$ gives a lower bound of
the number of eigenvalues for $U$:
\begin{equation} 
\label{eq014} 
\dim \ker (U - 1) + \dim \ker (U + 1)
 	\geq |{\rm ind}_\Gamma(U)|. 
\end{equation}
In particular, the equality in \eqref{eq014} holds
if $m_- = M_+ = 0$. 
We emphasize that 
$m_\pm - M_\pm$ depend on the choice of $\varGamma$
and so is ${\rm ind}_\varGamma(U)$,
while $m_\pm + M_\pm = \ker(U \mp 1)$ 
are independent of the choice of $\varGamma$. 
Example \ref{ex:4d} makes this evident. 
Therefore \eqref{eq014} motivates us to develop a way to know the index
without calculating the dimension of the kernels,
because it gives a sufficient condition 
for $U$ to have eigenvalues $\pm 1$.  
This possibility is explored in a forth coming paper \cite{SuTa}. 
As expected, we can prove that
${\rm ind}_\varGamma(U)$ 
is invariant under compact perturbations
if $H = Q^2$ is Fredholm.  
Thus we see that the eigenstates corresponding to $\pm 1$
for chiral symmetric quantum walks 
are robust against perturbations. 
This phenomena can be interpreted as a topological protection
of bound states (see Gesztesy and Simon \cite{GS}, 
where the invariance of the Witten index against compact perturbation
was called topological invariance). 
Therefore a nonzero index ${\rm ind}_\varGamma(U) \not=0$ 
can mathematically guarantee the existence of topologically protected 
bound states as found in \cite{AsOb13, KBFRBKAD, KRBD}. 
Such bound states are also expected to be localized at boundaries. 
Actually, in \cite{FFS1} we proved exponential decay of 
bound states in the birth eigenspaces.

\subsection{Comparison with related work} 
Avron et al. \cite{ASS1} defined an index for a Fredholm pair 
$(P_1, P_2)$ of two projections $P_1$ and $P_2$ so that
${\rm index}(P_1, P_2) 
= \dim \ker (P_1-P_2-1) - \dim \ker (P_1- P_2 + 1)$. 
This inspires us to introduce the following terminology
in order to give a criterion for the Fredholmness of
superhamiltonians. 
For two unitary involutions $\varGamma$ and $C$,
we say that $(\varGamma, C)$ is a Fredholm pair
if $H = Q^2$ is Fredholm with $Q$ defined in \eqref{def_Q}. 
Let $\varGamma_+$ and $C_\pm$ be the projections onto
$\ker (\varGamma-1)$ and $\ker (C \mp 1)$. 
The index for two projections and the index we consider 
in this paper are related as follows. 
\begin{equation}
\label{eq:inds} {\rm ind}_\varGamma(U)
	= {\rm index}(\varGamma_+, C_+) 
		+ {\rm index}(\varGamma_+,C_-). 
\end{equation}
Avron et al applied their index to study the charge deficiency \cite{ASS2},
in which they took two projections as $P_1$ and $P_2:=W P_1 W^*$
with $W$ a unitary operator and obtained
\[ {\rm Index}(P_1, P_2) = {\rm Tr}(\{[P_1,W]W^*\}^{2n+1}) \]
whenever $\{[P_1,W]W\}^{2n+1}$ is trace class. 
To define an index ${\rm ind} (U)$ 
for a one-dimensional quantum walk with an evolution $U$,
Gross et al \cite{Gross} employed the above formula 
with $n=0$, $W = U$, and
$P_1 = P$ the projection onto the half line,  
{\it i.e.}, ${\rm ind} (U) = {\rm Tr} ([P,U] U^*)$. 
Usually, standard one-dimensional quantum walks
have evolution operators of the form $U = SC$,
where $S$ is a shift operator and $C$ is a coin operator defined by 
a multiplication operator by $C(x) \in U(2)$. 
In such a case, the above index defined by Gross et al. 
cannot give different indices for different coins, 
because ${\rm Tr} ([P,U] U^*) = {\rm Tr}([P,S]S^*)$.
In particular, if 
$S = S_{\rm ss} :=\begin{pmatrix} p & \bar q L^* \\ q L &  -p \end{pmatrix}$
with $p^2 + |q|^2 = 1$ ($p \in \mathbb{R}$, $q \in \mathbb{C}$)
and $L$ the left-shift operator, 
a direct calculation yields ${\rm ind}(U) = {\rm Tr}([P,S]S^*) = 0$. 
In \cite{Ced1, Ced2}, Cedzich et al. dealt with indices 
defined by means of ${\rm Im} U:= (U- U^*)/2i$, 
which is equal to our supercharge $Q$, because 
$[\varGamma, C] =  \varGamma C - C \varGamma = U- U^*$.  
However, their construction and proofs seem to depend
on the one-dimensionality.  They did not obtain 
the formula \eqref{eq013} and did not mention supersymmetry. 

Because $S_{\rm ss}$ is a unitary involution, 
all one-dimensional quantum walks given by evolution operators
$U = S_{\rm ss} C$ with $C(x)$ unitary involution matrices 
are typical examples of index theory developed in this paper.  
This model includes 
all translation invariant standard one-dimensional quantum walks 
(even if $C(x)$ is not an involution) and 
Kitagawa's one-dimensional quantum walks \cite{KRBD}
(see \cite{FFS1, FFSd} for more details).  
We calculate the indices for such walks and give a trace formula
in a companion paper \cite{SuTa}.
Our framework also covers
multi-dimensional split-step quantum walks \cite{FFSd}, 
Grover's search algorithm (see Section \ref{subsec:GA}), 
the Grover walks on graphs (see Section \ref{subsec:GW}),  
the (twisted) Szegedy walks \cite{HKSS14}, 
and the Staggered quantum walks \cite{Po16}. 
 
This paper is organized as follows. 
Section 2 is devoted to defining the index for unitary operators. 
To this end, we study the relation between chiral symmetry 
for a unitary operator and supersymmetry for a pair of unitary involutions. 
In Section 3, we formulate the index formula \eqref{eq013} 
in terms of the spectral mapping theorem for pairs of unitary involutions. 
Here we also prove several properties for the index formula. 
In Section 4, we prove the index formula. 
We close this paper with three examples. 
In  Subsection 5.1, we give finite dimensional toy models.
In Subsection 5.2, we calculate the index for a unitary operator
that appears in Grover's search algorithm. 
Finally, we consider the Grover walks on graphs in Subsection 5.3.

\section{Chiral symmetry and supersymmetry}
\label{sec:CSSS}
Throughout this paper,
we assume that all Hilbert spaces are separable.  
We say that an operator $X$ is an involution
if $X^2=1$. The following is standard. 
\begin{remark}
\label{rem:USI}
If an operator $X$ has any two of the following three properties,
then it has all three properties:
(1) $X$ is self-adjoint, {\it i.e.}, $X^*=X$;
(2) $X$ is unitary, {\it i.e.}, $X^*=X^{-1}$; 
(3) $X$ is involutory, {\it i.e.}, $X^2=1$. 
\end{remark}
\subsection{Chiral symmetry}
\begin{definition}
Let $U$ be a unitary operator on a Hilbert space $\mathcal{H}$.
Then we say that $U$ has chiral symmetry 
if there exists a unitary involution 
$\varGamma$ on $\mathcal{H}$  such that
\begin{equation} 
\label{CS}
\varGamma U \varGamma = U^{-1}. 
\end{equation}
\end{definition}
\begin{lemma}
\label{lem:CSGC}
Let $U$ be a unitary operator on a Hilbert space $\mathcal{H}$.
The following are equivalent. 
\begin{itemize}
\item[(1)] $U$ has chiral symmetry.
\item[(2)] $U$ is a product of two unitary involutions. 
\end{itemize}
In particular, if $U$ satisfies \eqref{CS} with a unitary involution 
$\varGamma$, then $C:=\varGamma U$ is a unitary involution and
\begin{equation} 
\label{GC}
U = \varGamma C. 
\end{equation}
\end{lemma}
\begin{remark}
\label{rem:decomp}
The product decomposition of (2) is not necessary unique.
In fact, if $U_1$ and $U_2$ are unitary involutions,
then $-U_1$ and  $-U_2$ are also unitary involutions 
and $U_1 U_2 = (-U_1)(-U_2)$. 
See also Example \ref{ex:4d}. 
\end{remark}
\begin{proof}[Proof of Lemma \ref{lem:CSGC}]
Let $U$ satisfy \eqref{CS}. 
Then $C=\varGamma U$ is unitary. 
From \eqref{CS}, $C^2 = (\varGamma U)^2 = 1$. 
Because $\varGamma$ is a unitary involution,
$U = \varGamma^2 U = \varGamma \cdot (\varGamma U) = \varGamma C$.  
Hence (1) implies (2). 

Conversely, suppose that $U$ satisfies \ref{GC} and
$\varGamma$ and $C$ are unitary involutions. 
Then $\varGamma U \varGamma = C \varGamma = U^{-1}$.
Hence (2) implies (1). 

Thus we have the desired conclusion.
\end{proof}
\subsection{Supersymmetry}
Let $U$ have chiral symmetry 
with a unitary involution $\varGamma$ satisfying \eqref{CS}
and set $C = \varGamma U$. 
Then 
\[ R:= \frac{1}{2} \{ \varGamma, C \},
\quad 
Q:= \frac{1}{2i} [ \varGamma, C]
\]
are self-adjoint.  
\begin{lemma}
\label{lem:comm}
Let $U$ and $\varGamma$ be as stated above.
Then 
\begin{itemize}
\item[(1)] $[\varGamma,R]=0$,
\item[(2)] $\{\varGamma,Q\}=0$,
\end{itemize}
where $[X,Y]:=XY-YX$ and $\{X,Y\}:=XY+YX$. 
\end{lemma}
\begin{proof}
Because $C = \varGamma U$, we observe that
\[ R = {\rm Re} U := \frac{U+U^*}{2},
\quad 
Q = {\rm Im} U := \frac{U-U^*}{2i}.
\]
By \eqref{CS}, $\varGamma U = U^* \varGamma$
and $U \varGamma = \varGamma U^*$.  
Hence, $\varGamma R = R \varGamma$ and $\varGamma Q = -Q \varGamma$.
This proves (1) and (2).  
\end{proof}
From Remark \ref{rem:USI},
the spectrum of $\varGamma$ is $\sigma(\varGamma) = \{1, -1\}$
and the spectral decomposition of $\varGamma$ is
\[ \varGamma = \varGamma_+ - \varGamma_-, \]
where $\varGamma_\pm = (1 \pm \varGamma)/2$
is the projection onto $\ker(\varGamma \mp 1)$. 
With the identification $\mathcal{H} = {\rm Ran} \varGamma_+ \oplus {\rm Ran} \varGamma_-$,
$\varGamma$ is written as 
\[ \varGamma = \begin{pmatrix} 1 & 0 \\ 0 & -1 \end{pmatrix}. \]
With this notation, (2) of Lemma \ref{lem:comm} yields
\begin{equation}
\label{rep.Q} 
Q = \begin{pmatrix} 0 & \alpha^* \\ \alpha & 0\end{pmatrix}, 
\end{equation}
where $\alpha = \varGamma_- Q \varGamma_+$ 
is an operator from ${\rm Ran} \varGamma_+ \to {\rm Ran} \varGamma_-$. 
We set $H = Q^2$ and write
\[ H = \begin{pmatrix} H_+ & 0 \\ 0 & H_- \end{pmatrix}, \]
where $H_+ = \alpha^* \alpha$ and $H_- = \alpha \alpha^*$. 
In the context of supersymmetry,
$Q$ is called a {\it supercharge} and $H$ is a {\it superhamiltonian}
\cite{Witten} (see also \cite{Thaller}). 
The Witten index of $H$ is defined as
\[ \Delta(H) = \dim \ker H_+ - \dim \ker H_-.  \]
\subsection{Index of a Fredholm pair}
In this paper, we introduce an index for a pair $(U, \varGamma)$
of a unitary operator $U$ and a unitary involution $\varGamma$
satisfying \eqref{CS}. 
Before that, inspired by \cite{ASS1}, 
we introduce the following terminology.
We say that a pair $(C_1, C_2)$ of two unitary involutions 
is a Fredholm pair if $([C_1, C_2]/2i)^2$ is Fredholm. 
By definition, the pair $(\varGamma, C)$ of 
unitary involutions with  $C:= \varGamma U$
is a Fredholm pair if and only if $H = Q^2$ is Fredholm.

\begin{definition}
\label{def:ind}
Let $U$ and $\varGamma$ satisfy \eqref{CS}
and let $\alpha$ be as stated above. 
We say that $(U, \varGamma)$
is a Fredholm pair if $\alpha$ is Fredholm,
{\it i.e.}, $\dim \ker \alpha < \infty$, $\dim \ker \alpha^* < \infty$,
and ${\rm Ran}(\alpha)$ is closed.   
In this case, the index of the pair $(U,\varGamma)$ is defined by
\begin{equation}
\label{eq.ind}
{\rm ind}_\varGamma (U) = {\rm index} (\alpha), 
\end{equation}
where 
${\rm index} (\alpha):= \dim \ker \alpha - \dim \ker \alpha^*$ 
is the Fredholm index of $\alpha$. 
\end{definition}
\begin{proposition}
\label{pr:fred}
Let $U$ be unitary and 
$\varGamma$, $\varGamma^\prime$ be unitary involutions
satisfying $\varGamma U \varGamma = U^{-1}$ 
and $\varGamma^\prime U \varGamma^\prime = U^{-1}$. 
If $(U,\varGamma)$ is a Fredholm pair,
then so is $(U,\varGamma^\prime)$. 
\end{proposition}
\begin{proof}
Observe that the operator $\alpha$ is Fredholm if and only if
\begin{equation}
\label{eq:hfred} 
\dim \ker H_\pm < \infty \quad \mbox{and} \quad 
\inf \sigma(H_+)\setminus \{0\} > 0. 
\end{equation}
Moreover, \eqref{eq:hfred} is equivalent to saying that
$H$ is Fredholm
(see \cite{ASS1, Thaller} for more details).  
Because $H = ({\rm Im} U)^2$ is independent of the choice 
of $\varGamma$,
we obtained the desired assertion.
\end{proof}
\begin{remark}
\label{rem:ind}
Because from the above proof, $\alpha$ is Fredholm if and only if
$H$ is Fredholm, 
$(U,\varGamma)$ is a Fredholm pair if and only if so is $(\varGamma, C)$.
Therefore we henceforth only consider the Fredholmness of pairs 
$(U,\varGamma)$ of a unitary operator $U$ 
and a unitary involution $\varGamma$.   
Moreover, we observe from \eqref{eq:hfred} that 
if $(U,\varGamma)$ is a Fredholm pair, 
then ${\rm index}(\alpha) = \Delta(H)$. 
Because 
the Hamiltonian $H$ is independent of the choice of $\varGamma$,
one may feel that the index is also independent of 
the choice of $\varGamma$. 
However, the definition of the Witten index 
depends on the choice of $\varGamma$,
because $H_\pm$ are determined by $\varGamma$. 
As mentioned in Remark \ref{rem:decomp}, 
the decomposition $U=\varGamma C$ 
by two unitary involutions $\varGamma$ and $C$ is not necessary unique.
Hence, it is possible that there are two unitary involutions 
$\varGamma_i$  such that
$\varGamma_i U \varGamma_i = U^{-1}$ ($i=1,2$).  
Indeed, Example \ref{ex:4d} reveals that 
there are Fredholm pairs $(U,\varGamma)$ 
and $(U,\varGamma^\prime)$
such that 
${\rm ind}_{\varGamma^\prime}(U) \not= {\rm ind}_{\varGamma}(U) $.
That is why we define an index
for a pair $(U,\varGamma)$ and not for $U$ itself.
\end{remark}
\section{Index formula}
\label{sec:ind}
Throughout this section, 
we assume that $U$ is a unitary operator 
on a Hilbert space $\mathcal{H}$ and it has chiral symmetry.  
Then Lemma \ref{lem:CSGC} says that 
$U$ is written as a product of two unitary involutions 
$\varGamma$ and $C$, {\it i.e.}, $U = \varGamma C$. 
Without loss of generality, 
we can suppose that $\ker(C-1) \not=\{0\}$,
because if $\ker(C-1) =\{0\}$, 
then $-C= 1$ and $-\varGamma$ are unitary involutions
and $U = (-\varGamma)(-C)$ is a product of two unitary involutions. 

In this section,  
we give an explicit expression for 
${\rm ind}_\varGamma(U)$ defined in \eqref{eq.ind}. 
To this end, we review a previous result \cite{HSS, SS}
on a spectral mapping theorem for a product of 
two unitary involutions. 
\begin{theorem}[\cite{HSS,SS}]
\label{thm:smt}
Let $\mathfrak{H}$ and $\mathfrak{K}$ 
be Hilbert spaces. 
Let $\gamma$ be a unitary involution on $\mathfrak{H}$ and 
$\partial:\mathfrak{H} \to \mathfrak{K}$ be a coisometry,
{\it i.e.}, $\partial \partial^* = 1$ on $\mathfrak{K}$.
\begin{itemize}
\item[(i)] $\tau := \partial \gamma \partial^*$ is 
bounded and self-adjoint on $\mathfrak{K}$
with $\|\tau\| \leq 1$. 
\item[(ii)] $u:= \gamma (2\partial^*\partial -1)$ is unitary and
\[ \sigma_\sharp(u) 
	= \varphi^{-1}(\sigma_\sharp(\tau)),
	\quad \sharp = \mbox{c, ac, sc}, 
\]
where $\varphi(z) = (z+z^{-1})/2$.
\item[(iii)] For $\lambda \in \sigma_{\rm p}(u)$,
\[ \dim \ker(u-\lambda)
	= \begin{cases} 
		\dim \ker (\tau -\varphi(\lambda)), & \lambda \not= \pm 1 \\
		m_\pm + M_\pm, & \lambda = \pm 1,
	\end{cases} \]
where $m_\pm:= \dim \ker (\tau \mp 1)$
and  $M_\pm := \dim \ker (\gamma \pm 1) \cap \ker \partial$. 
\end{itemize}
\end{theorem} 
We now give different expressions for $m_\pm$ and $M_\pm$. 
\begin{corollary}
\label{cor:new}
Let $\gamma$, $\partial$, and $m_\pm$ be as stated in Theorem \ref{thm:smt}. 
\begin{align} 
& \label{eq:m+} 
	m_\pm = \dim \ker (\gamma \mp 1) \cap \ker (c - 1), \\
& \label{eq:M+} 
	M_\pm = \dim \ker (\gamma \pm 1) \cap \ker (c + 1), 	
\end{align}
where $c:=2\partial^* \partial - 1$.  
\end{corollary}
\begin{proof}
 Because $\ker \partial = \ker (c+1)$, 
 \eqref{eq:M+} is obtained from
 \begin{equation}
 \label{eq:birtexp}
 \ker (\gamma \pm 1) \cap \ker (c + 1)
 = \ker (\gamma \pm 1) \cap \ker \partial.
 \end{equation}  
 We prove \eqref{eq:m+}. 
 Because $\partial^*$ is a bijection 
 from $\ker (\tau \mp 1)$ to $\partial^* \ker (\tau \mp 1)$, 
 we need only prove 
 \begin{equation}
 \label{eq:inhexp} 
 \partial^* \ker (\tau \mp 1) 
 = \ker (\gamma \mp 1) \cap \ker (c-1). 
 \end{equation} 
 Because ${\rm Ran}(\partial^*) = \ker( c-1)$,
 \begin{align*} 
 \ker (\gamma \mp 1) \cap \ker (c-1)
 = \{ \partial^* f \mid \gamma \partial^* f = \pm \partial^* f\}.
 \end{align*}
 If $\gamma \partial^* f = \pm \partial^* f$, then
 $\tau f = \partial \gamma \partial^* f = \pm f$. 
 Hence $f \in \ker (\tau \mp 1)$. 
 Conversely, if $f \in \ker (\tau \mp 1)$, then  
 \[ \langle \partial^* f, S \partial^* f \rangle
 	=\langle f, Tf \rangle =  \pm \|f\|^2  = \pm \|\partial^* f\|^2. \]
 Subtracting the right-hand side from the left-hand side
 yields $\|(\gamma \mp 1)  \partial^* f\|^2 = 0$, because $(1 \pm \gamma)/2$
 is a projection. Hence $ \partial^* f \in \ker(\gamma \mp 1)$
 and \eqref{eq:inhexp} is proved. 
 This concludes the desired assertion.  
\end{proof}

To apply Theorem \ref{thm:smt} and Corollary \ref{cor:new} 
to the operator $U=\varGamma C$,
we will represent $C$ using a coisometry.  
Let $\{\chi_j \}_{j \in V}$ be a CONS of $\ker(C-1)$,
where $V$ is a countable set. 
We use $\mathcal{K}$ to denote the Hilbert space $\ell^2(V)$
of square summable functions on $V$.  
We introduce 
an operator $d: \mathcal{H} \to \mathcal{K}$ as follows.  
For $\psi \in \mathcal{H}$,
$d\psi \in \mathcal{K}$ is defined as a function on $V$ such that 
\[ (d\psi)(j)  = \langle \chi_j, \psi \rangle,
	\quad j \in V, \]
where  $\langle \cdot, \cdot \rangle$ on the right-hand side is 
the inner product on $\mathcal{H}$.  
The Bessel inequality guarantees the bondedness of $d$.
The following lemma is straightforward.
For a proof, the reader can consult \cite{HSS}.  
\begin{lemma}
\label{lem:coiso}
Let $d$ be as stated above. 
\begin{itemize}
\item[(i)] The operator $d$ is a coisometry, 
{\it i.e.}, $dd^* = 1$ on $\mathcal{K}$. 
\item[(ii)] The adjoint $d^*:\mathcal{K} \to \mathcal{H}$ is
an isometry and is given by 
\[ d^*f = \sum_{j \in V}f(j) \chi_j,
	\quad  \mbox{$f \in \mathcal{K}$}. \] 
\item[(iii)] $d^*d = \sum_{j \in V} |\psi_j\rangle \langle \psi_j|$
is the projection onto $\ker(C-1)$. 
\item[(iv)] $C = 2d^*d - 1$. 
\end{itemize}
\end{lemma}  
Because, from the above lemma, 
any chiral symmetric unitary operator $U$ can be written as
$U = \varGamma(2d^*d-1)$ with 
a unitary involution $\varGamma$ and a coisometry $d$,
Theorem \ref{thm:smt} is applicable for $U$. 
Let $m_\pm = \dim \ker (T \mp 1)$ 
and $M_\pm = \dim \mathcal{B}_\pm$,
where 
$T = d \varGamma d^*$ is called the discriminant of $U$
and $\mathcal{B}_\pm = \ker (\varGamma \pm 1) \cap \ker d$
is called the birth eigenspaces \cite{Se13, MOS}. 
From the proof of Corollary \ref{cor:new},
\begin{equation}
\label{eq:brt} 
\mathcal{B}_\pm = \ker(\varGamma \pm 1) \cap \ker (C + 1). 
\end{equation} 
In this paper, we introduce inherited eigenspaces 
$\mathcal{T}_\pm$ by
\begin{equation}
\label{eq:inh}
\mathcal{T}_\pm = \ker (\varGamma \mp 1) \cap \ker (C-1). 
\end{equation}
Corollary \ref{cor:new} says that
$m_\pm = \dim \mathcal{T}_\pm$. 
\begin{remark}
It is worthy noting that
the inherited eigenspaces and the birth eigenspaces 
can be represented as
\begin{align}
\label{eq:18:04}  
\mathcal{T}_\pm = d^* \ker (T \mp 1), 
\quad
\mathcal{B}_\pm = \ker( \varGamma \pm 1) \cap \ker d
\end{align} and
\begin{equation} 
\label{eq:egnpm1}
\ker (U \mp 1) = \mathcal{T}_\pm \oplus \mathcal{B}_\pm. 
\end{equation}
Here \eqref{eq:18:04} has already been proved 
in \eqref{eq:inhexp} and \eqref{eq:birtexp}.
For the proof of \eqref{eq:egnpm1}, the reader  can consult \cite{SS}. 
\eqref{eq:brt}, \eqref{eq:inh}, and \eqref{eq:egnpm1}
prove \eqref{eq:inds}.  

\end{remark}
In terms of the spectral mapping theorem,
we obtain the following index theorem. 
\begin{theorem}
\label{thm:ind}
Let $U$, $\varGamma$, 
$T=d \varGamma d^*$, $m_\pm$, and $M_\pm$ be as stated above. 
\begin{itemize}
\item[(i)] $(U, \varGamma)$ is a Fredholm pair if and only if $1-T^2$ is Fredholm and $M_\pm < \infty$.  
\item[(ii)] If $(U, \varGamma)$ is a Fredholm pair,
\begin{equation}
\label{eq:ind}
{\rm ind}_\varGamma(U)
=(M_- - m_-) - (M_+ - m_+). 
\end{equation}
In particular, 
\begin{equation}
\label{eq:lower}  
\dim \ker (U-1) + \dim \ker (U+1) 
	\geq |{\rm ind}_\varGamma(U)|, 
\end{equation}
where the equality holds if $m_-=0$ and $M_+=0$. 
\item[(iii)] If $(U, \varGamma)$ and $(U^\prime, \varGamma)$
are Fredholm pairs and $U^\prime - U$ is compact,
then
\[ {\rm ind}_\varGamma(U^\prime) =  {\rm ind}_\varGamma(U). \]
\end{itemize}
\end{theorem}
We postpone the proof  until the next section. 
In what follows, we give several corollaries of Theorem \ref{thm:ind}. 
\begin{corollary}
\label{cor:T<1}
Let $U$, $\varGamma$, $T$, and $M_\pm$ be as stated above. 
If $\|T\| < 1$ and $M_\pm < \infty$, 
then $(U,\varGamma)$ is a Fredholm pair,
$m_\pm=0$, and
\[ {\rm ind}_\varGamma (U) = M_- - M_+ \]
\end{corollary}  
\begin{proof}
By assumption, there exists a positive $\epsilon$ such that 
$\|T\| = 1-\epsilon$.  
This implies $1-T^2 \geq \epsilon$.  
Hence $1-T^2$ is Fredholm and $m_\pm = 0$. 
Applying Theorem \ref{thm:ind}, 
we have the desired assertion. 
\end{proof}
The following corollary indicates unitary invariance of the index. 
\begin{corollary}
Let $U$ and $\varGamma$ be as above. 
Let $V$ be a unitary operator from $\mathcal{H}$ 
to a Hilbert space $\mathcal{H}^\prime$
and set $U^\prime = VUV^{-1}$
and $\varGamma^\prime = V \varGamma V^{-1}$. 
If $(U, \varGamma)$ is a Fredholm pair,
then so is $(U^\prime, \varGamma^\prime)$ and
\begin{equation} 
\label{eq:inv}
{\rm ind}_{\varGamma^\prime} (U^\prime) 
	= {\rm ind}_\varGamma(U). 
\end{equation}
In particular, $(U^{-1},\varGamma)$ is a Fredholm pair and
\begin{equation} 
\label{eq:inverse}
{\rm ind}_\varGamma(U^{-1}) = {\rm ind}_\varGamma(U). 
\end{equation}
\end{corollary}
\begin{proof}
We first prove \eqref{eq:inv}. 
By Lemma \ref{lem:coiso}, $\varGamma U = C = 2d^*d-1$
and $d$ is a coisometry.  
Then $d^\prime := d V^{-1}$
is a coisometry from $\mathcal{H}^\prime$ to $\mathcal{K}$. 
Indeed, 
$d^\prime (d^\prime)^* = d V^{-1} \cdot V d^* = d d^* = 1$
on $\mathcal{K}$. 
Observe that $U^\prime = \varGamma^\prime C^\prime$
is a product of two unitary involutions
$\varGamma^\prime$ 
and $C^\prime := VCV^\prime$. 
Moerover,  
$C^\prime = V (2d^*d -1) V^{-1} = 2(d^\prime)^* d^\prime -1$.
The discriminant $T^\prime$ of $U^\prime$ is equal to $T$,
because it is given by
\[ T^\prime = d^\prime \varGamma (d^\prime)^*
= d V^{-1} \cdot V\varGamma V^{-1} \cdot V d^*
= d \varGamma d^* = T. \] 
Let $M_\pm$ be the dimension of
$\mathcal{B}_\pm^\prime = \ker(\varGamma^\prime \pm 1) \cap \ker d^\prime$
and $m_\pm^\prime = \ker(T^\prime \mp 1)$.
Since $T^\prime = T$, $m_\pm^\prime = m_\pm$. 
By definition, 
\begin{align*}
\mathcal{B}_\pm^\prime
& = \{ \psi \mid d^\prime \psi = 0, \ (\varGamma^\prime \pm 1) \psi = 0 \} \\
& = \{ \psi \mid d (V^{-1} \psi) = 0, \ (\varGamma \pm 1) (V^{-1} \psi) = 0 \} \\
& = V \mathcal{B}_\pm,
\end{align*} 
which implies $M_\pm^\prime = M_\pm$. 
Thus we find
from Theorem \ref{thm:ind}
that if $(U, \varGamma)$ is a Fredholm pair,
then so is $(U^\prime, \varGamma)$ and 
\begin{align*} 
{\rm ind}_{\varGamma^\prime}(U^\prime)
& = (M_-^\prime - m_-^\prime) - 	(M_+^\prime - m_+^\prime) \\
& = (M_- - m_-) - 	(M_+ - m_+)
	= {\rm ind}_\varGamma(U). 
\end{align*} 
Hence, \eqref{eq:inv} is proved.

We next prove \eqref{eq:inverse}. 
Let $V = \varGamma$. 
We obtain $U^\prime = \varGamma U \varGamma = U^{-1}$,
$\varGamma^\prime = \varGamma$, 
${\rm ind}_{\varGamma^\prime}(U^\prime) 
= {\rm ind}_\varGamma(U^{-1})$. 
Applying \eqref{eq:inv} yields \eqref{eq:inverse}. 
The proof is completed. 
 \end{proof}
\begin{remark}
In general, the Fredholm index satisfies
\[ {\rm Index}(A^*) = - {\rm Index}(A). \]
Hence, \eqref{eq:inverse} may seem strange. 
It should be noted that ${\rm ind}_\varGamma(U)$ 
(or equivalently the Witten index $\Delta(H)$) is defined 
through the Hamiltonian. 
Because the Hamiltonian $H$ for $U$ 
is equal to $H^\prime$ for $U^{-1}$,
\eqref{eq:inverse} is correct. 
Indeed, 
$H^\prime = ({\rm Im}(U^{-1}))^2 = ({\rm Im} U)^2 = H$. 
As seen below, even if $U$ and $U^\prime$ have the same Hamiltonian,
it is possible to have different indices. 
See also Remark \ref{rem:ind} and Example \ref{ex:4d}. 
\end{remark} 
Using \eqref{eq:brt} and \eqref{eq:inh},
we obtain the following.
\begin{corollary}
\label{cor:-U}
Let $U$ and $\varGamma$ be as stated above. 
If $(U,\varGamma)$ is a Fredholm pair, 
then so are $(-U,\varGamma)$ and $(U,-\varGamma)$ 
and
\[ {\rm ind}_\varGamma (-U) =  {\rm ind}_\varGamma(U),
	\quad
	{\rm ind}_{-\varGamma}(U) = - {\rm ind}_\varGamma(U). 
	\]
\end{corollary}
\begin{proof}
Let $H$ be the Hamiltonian for $U$.  
Then the Hamiltonian for $-U$ is equal to $H$,
because $({\rm Im} (-U))^2 = ({\rm Im} (U))^2 = H$. 
Hence, from an argument similar to the proof of Proposition \ref{pr:fred},
if $(U, \varGamma)$ is a Fredholm pair, then so are 
$(-U,\varGamma)$ and $(U,-\varGamma)$.  

We now write $\mathcal{B}_\pm(U,\varGamma)$ 
and $\mathcal{T}_\pm(U,\varGamma)$ 
for the birth eigenspaces \eqref{eq:brt} 
and the inherited eigenspaces \eqref{eq:inh}
to make the dependence on $U$ and $\varGamma$ explicit. 
For the pair $(-U, \varGamma)$, $U$ is decomposed 
into $- U = \varGamma(-C)$,
\eqref{eq:brt} and \eqref{eq:inh} 
say that
\begin{align*}
& \mathcal{B}_\pm(-U, \varGamma)
= \ker (\varGamma \pm 1) \cap  \ker (-C + 1)
= \mathcal{T}_\mp(U,\varGamma), \\
& \mathcal{T}_\pm(-U, \varGamma)
= \ker (\varGamma \mp 1) \cap  \ker (-C - 1)
= \mathcal{B}_\mp(U,\varGamma),
\end{align*}
which, combined with \eqref{eq:ind}, imply that
\begin{align*}
{\rm ind}_\varGamma(-U)
&= (\dim \mathcal{T}_+(U,\varGamma)
	- \dim \mathcal{B}_+(U,\varGamma)) \\
& \quad -  (\dim \mathcal{T}_-(U,\varGamma)
	- \dim \mathcal{B}_-(U,\varGamma))  \\
& = {\rm ind}_\varGamma (U).	
\end{align*}
Similarly, $U = (-\varGamma)(-C)$ implies that
\begin{align*}
\mathcal{B}_\pm(U, -\varGamma)
& = \ker (- \varGamma \pm 1) \cap  \ker (-C + 1) \\
& = \ker (\varGamma \mp 1) \cap  \ker (C - 1)
= \mathcal{T}_\pm(U,\varGamma), \\
\mathcal{T}_\pm(U, -\varGamma) 
& = \ker (- \varGamma \mp 1) \cap  \ker (-C - 1) \\
& = \ker (\varGamma \pm 1) \cap  \ker (C + 1) 
= \mathcal{B}_\pm(U,\varGamma),
\end{align*}
which yields 
\begin{align*}
{\rm ind}_{-\varGamma}(U)
&= (\dim \mathcal{T}_-(U,\varGamma)
	- \dim \mathcal{B}_-(U,\varGamma)) \\
& \quad -  (\dim \mathcal{T}_+(U,\varGamma)
	- \dim \mathcal{B}_+(U,\varGamma))  \\
& = - {\rm ind}_\varGamma (U).	
\end{align*}
This completes the proof. 
\end{proof}
We conclude this section
with $\mathcal{H}$ finite dimensional. 
\begin{corollary}
\label{cor:fin}
Let $U$, $\varGamma$, $T$, $m_\pm$, and $M_\pm$
be as stated in Theorem \ref{thm:ind} 
and suppose that $\dim \mathcal{H} < \infty$. 
Then $(U,\varGamma)$ is a Fredholm pair 
and \eqref{eq:ind} holds. 
Moreover, if $\dim \ker(\varGamma+1) = \dim \ker(\varGamma-1)$,
then
\[ {\rm ind}_\varGamma (U) = 0. \]
\end{corollary}  
\begin{proof}
Since $\dim \mathcal{H} < \infty$,
$\alpha$ is automatically Fredholm. 
Hence, from Theorem \ref{thm:ind},
\eqref{eq:ind} holds. 
If $n=\dim \ker(\varGamma+1) = \dim \ker(\varGamma-1)$,
then $\alpha$ is viewed as a square matrix of order $n$
and hence $\dim \ker \alpha =  \dim \ker \alpha^*$. 
By definition, ${\rm ind}_\varGamma(U) = \ker \alpha - \ker \alpha^*=0$. 
This completes the proof.
\end{proof}

 \section{Proof of the index formula} 
In this section, we prove Theorem \ref{thm:ind}.
Throughout this section,
we assume that  a unitary operator $U$ and 
a unitary involution $\varGamma$
satisfy  $\varGamma U \varGamma = U^{-1}$.    
We use the notations in Sections \ref{sec:CSSS} and \ref{sec:ind}. 
We prove (i)-(ii) of Theorem \ref{thm:ind} in Subsections \ref{subsec:fred}-%
\ref{subsec:inv}. 
 \subsection{Fredholmness}
 \label{subsec:fred}
In this subsection,
we prove (i) of Theorem \ref{thm:ind},
{\it i.e.}, $(U, \varGamma)$ is a Fredholm pair 
if and only if $1-T^2$ is Fredholm and $M_\pm < \infty$. 
\begin{proof}[Proof of Theorem \ref{thm:ind} (i)]
By Remark \ref{rem:ind}, 
$(U, \varGamma)$ is a Fredholm pair if and only if $H$ is Fredholm.
Thus  we find that the following proposition concludes the proof. 
\end{proof} 
\begin{proposition}
\label{pr:ht}
The following are equivalent.
\begin{itemize}
\item[(i)] $H$ is Fredholm.
\item[(ii)] $1-T^2$ is Fredholm and $M_\pm < \infty$. 
\end{itemize}
\end{proposition} 
\begin{proof}
By \cite[Theorem 4.1 and Proposition 5.3]{SS},  
$\mathcal{H}$ is unitarily equivalent to
$\ker(1-T^2)^\perp \oplus \ker(1-T^2)^\perp \oplus \ker (1 - U)
\oplus \ker (1 + U)$
and with this identification
\[ U \simeq e^{i \arccos T} 
	\oplus e^{-i \arccos T} \oplus 1 \oplus (-1). \]
Because $e^{\pm \arccos T} = T \pm \sqrt{1-T^2}$,
\[ Q = {\rm Im} U 
	\simeq \sqrt{1-T^2} \oplus  \sqrt{1-T^2} \oplus 0 \oplus 0 \]
and hence
\[ H = Q^2 \simeq 1-T^2 \oplus  1-T^2 \oplus 0 \oplus 0. \]
Therefore, 
\begin{equation}
\label{eq:htspec} \inf \sigma (H) \setminus \{0\}
	= \inf \sigma (1-T^2) \setminus \{0\}. 
\end{equation}
Because, by Theorem \ref{thm:smt},
$\dim \ker (1\mp U) 
= m_\pm + M_\pm$
and by definition, $m_+ + m_- =\ker (1 - T^2) $,
\begin{align}
\dim \ker H 
& = \dim \ker (1-U) \oplus  \ker (1+U) 
\notag \\
& = m_+   + M_+ + m_- + M_-
\notag \\
& = \dim \ker (1-T^2) + M_+ + M_-.
\label{eq:htker}
\end{align}
\eqref{eq:htspec} and \eqref{eq:htker}
conclude the desired assertion. 
\end{proof}
 \subsection{$\ker \alpha$ and $\ker \alpha^*$}
 \label{subsec:ind}
 In this section, we prove Theorem \ref{thm:ind} (ii),
 {\it i.e.}, if $(U,\varGamma)$ is a Fredholm pair, then
 \[ {\rm ind}_\varGamma(U) =(M_- - m_-) - (M_+ - m_+)
 	\tag{\ref{eq:ind}}.
 \] 
 \eqref{eq:lower} can be easily proved by the above equation. 
 \begin{proof}[Proof of Theorem \ref{thm:ind} (ii) ]
 Because the right-hand side of \eqref{eq:ind} is
 \[ 
  (m_+ + M_-) - (m_- + M_+)
  = (\dim \ker (1- T) + \dim \mathcal{B}_-)
 	- (\dim \ker (1+ T) + \dim \mathcal{B}_+), \]
it suffices to  prove
 \begin{align*}
 & \dim \ker \alpha = \dim \ker (1- T) + \dim \mathcal{B}_-, \\
 & \dim \ker \alpha^* = \dim \ker (1+ T) + \dim \mathcal{B}_+. 
 \end{align*}
 Because $d^*$ is a bijection 
 from $\dim \ker (1- T) = \dim d^* \ker (1- T)$,
 the following proposition prove the desired assertion.
 \end{proof}
 \begin{proposition}
 \label{pr:keralpha}
 \begin{itemize}
 \item[(i)] $\ker \alpha = d^* \ker (1- T) \oplus \mathcal{B}_-$.
 \item[(ii)] $\ker \alpha^* = d^* \ker (1 + T) \oplus \mathcal{B}_+$. 
 \end{itemize} 
 \end{proposition}

 The proof of Proposition \ref{pr:keralpha} splits into several lemmas. 
 \begin{lemma}
 \label{lem:QU}
 $\ker Q = \ker (1- U^2)$. 
 \end{lemma}
 \begin{proof}
 Supposing $\varphi \in \ker(1- U^2)$,
 we have $(U - U^*) \varphi = 0$ and hence
 $(1- U^2)\varphi = - U (U- U^*) \varphi = 0$. 
 Therefore, $\varphi \in \ker (1 - U^2)$. 
 
 Conversely, suppose that $\varphi \in \ker (1-U^2)$. 
 Then $U^2 \varphi = \varphi$ and hence $U \varphi = U^* \varphi$.
 Hence, $Q \varphi = (U- U^*)\varphi/ 2i = 0$. 
 This completes the proof.
 \end{proof}
 \begin{lemma}
 \label{lem:aqg}
 \begin{itemize}
 \item[(i)] $\ker \alpha = \ker Q \cap \ker(\varGamma - 1)$
 \item[(ii)] $\ker \alpha^* = \ker Q \cap \ker(\varGamma + 1)$
 \end{itemize}
 \end{lemma}
 \begin{proof}
 Since $\alpha = \varGamma_- Q \varGamma_+$ is an operator
 from ${\rm Ran} \varGamma_+$ to ${\rm Ran} \varGamma_-$,
 \[ \ker \alpha 
 	= \{ \varphi \in {\rm Ran} \varGamma_+ 
 	\mid \varGamma_- Q \varphi = 0 \}. \]
 Supposing that $\varphi \in \ker \alpha$,
 we have $(1- \varGamma) Q \varphi = 0$. 
 Because, by Lemma \ref{lem:comm}, 
 $Q$ anticommutes with $\varGamma$,  
 we obtain $Q \varphi = \varGamma Q \varphi =  - Q\varGamma \varphi$. 
 Hence, $Q(1+\varGamma) \varphi = 0$. 
 Because $\varphi \in {\rm Ran} \varGamma_+$,
 $Q \varphi = 0$.
 Thus, we see that $\varphi \in \ker Q \cap \ker (\varGamma - 1)$. 
 
 Conversely, suppose that $\varphi \in \ker Q \cap \ker (\varGamma - 1)$. 
 Then $\varphi \in {\rm Ran}( \varGamma_+)$ and 
 $Q \varphi = 0$. Hence, $\varphi \ker \alpha$. 
 Therefore (i) is proved.
 The same proof works for (ii).  
 \end{proof}
 \begin{lemma}
 \label{lem:UC}
 Let $C_\pm = (1 \pm C)/2$.  
 For any $\varphi \in \ker \alpha$,
 \begin{equation*}
 \label{eq:UP1} U C_\pm \varphi = \pm C_\pm \varphi.
 \end{equation*}
 \end{lemma}
 \begin{proof}
 Observe that $C_\pm$ is the projection onto $\ker(C \mp 1)$.
 Let $\varphi \in \ker \alpha$. 
 By Lemmas \ref{lem:QU} and \ref{lem:aqg},
 $\varphi$ belongs to $\ker (1-U^2)$ and $\ker ( \varGamma -1 )$. 
 Hence, $U^2 \varphi = \varphi$ and 
 \begin{equation}
 \label{eq:52}
  U\varphi = U^* \varphi = C \varGamma \varphi = C \varphi. 
 \end{equation}
 By \eqref{eq:52}, 
 \begin{equation}
 \label{eq:53}
 UC \varphi = U^2 \varphi = \varphi.
 \end{equation}
 By \eqref{eq:52} and \eqref{eq:53},
 \[ U \left(\frac{1 \pm C}{2} \right) \varphi
 	= \frac{C \pm 1}{2} \varphi, \]
 which proves the lemma. 
 \end{proof}
 We now prove Proposition  \ref{pr:keralpha},
 using Lemmas \ref{lem:QU}, \ref{lem:aqg}, and \ref{lem:UC}. 
 \begin{proof}[Proof of Proposition \ref{pr:keralpha}]
 Suppose that $\varphi \in \ker \alpha$. 
 With the decomposition 
 $\mathcal{H} = {\rm Ran} d^* \oplus \ker d$,
 we can write 
 \[ \varphi = d^* f + \varphi_0, \]
 where $f \in \mathcal{K}$, $\varphi_0 \in \ker d$. 
 Since $C_\pm$ is the projection onto $\ker (C \mp 1)$
 Lemma \ref{lem:coiso} says that $C_+ = d^* d$. 
 Hence, ${\rm Ran} d^* = \ker (C - 1)$ and 
 ${\ker} d = \ker (C + 1)$. 
 Because $d$ is a coisometry,
 \begin{align}
 & C_+ \varphi = d^* f, \label{3:02} \\
 & C_- \varphi = \varphi_0.  \label{3:03}
 \end{align}
 By \eqref{3:02},
 \begin{align*}
 \varGamma d^*f = \varGamma C_+ \varphi
 	= \varGamma C C_+ \varphi
 	= U C_+ \varphi = C_+ \varphi = d^* f,
 \end{align*}
 where we have used Lemma \ref{lem:UC} 
 in the second last equality. 
 Hence, 
 \[ Tf = d (\varGamma d^* f) = d (d^* f) = f \] 
 and therefore $f \in \ker (T-1)$. 
 Similarly, by \eqref{3:03},
 \begin{equation}
  \label{3:10}
 U \varphi_0 = U C_- \varphi = - C_- \varphi = - \varphi_0,
 \end{equation}
where we have used Lemma \ref{lem:UC} again
 in the second last equality. 
 Because $\varphi_0 \in \ker (C+1)$, 
 \begin{equation} 
 \label{3:11}
 U\varphi_0 = \varGamma C \varphi_0 = - \varGamma \varphi_0.
 \end{equation} 
 By \eqref{3:10} and \eqref{3:11},
 $\varGamma \varphi_0 = \varphi_0$.
 Hence, $\varphi_0 \in \ker (\varGamma -1) \cap \ker d = \mathcal{B}_-$.
 Thus we see that $\varphi \in d^* \ker (T=1) \oplus \mathcal{B}_-$. 
 
 Conversely, supposing that 
 $\varphi \in d^* \ker(T-1) \oplus \mathcal{B}_-$,
 we can write
 \begin{equation}
 \label{3:23} 
 \varphi = d^*f + \varphi_0, 
 \end{equation}
 where $f \in \ker (T-1)$ and $\varphi_0 \in \mathcal{B}_-$.
 We now claim that 
 \begin{equation}
 \label{3:25}
 d^*f \in \ker (\varGamma-1).
 \end{equation}
 Indeed, 
 an argument similar to the proof of Corollary \ref{cor:new} 
 yields 
 \eqref{3:25}. 
 Since $\varphi_0 \in \mathcal{B}_- \subset \ker (\varGamma -1)$,
 \eqref{3:23} and \eqref{3:25} imply $\varphi \in \ker(\varGamma-1)$. 
 We next prove that $\varphi \in \ker Q$. 
 Combining \eqref{3:25} and (iii) of Lemma \ref{lem:coiso}, 
 we have
 \[ 2i Q d^* f = (U - U^*) d^* f
 	= \varGamma d^*f - C \varGamma d^*f
 	= (1- C) d^*f =  2 C_- d^* f = 0. \]
 Hence, $d^* f \in \ker Q$.
 Similarly, using $\varphi_0 \in \ker ( C+1)$, we have
 \[ 2i Q\varphi_0 = (U - U^*) \varphi_0 
 	=  - \varGamma \varphi_0 - C \varGamma \varphi_0
 	= - C_+ \varphi_0 = 0. \]
 Hence, $\varphi_0 \in \ker Q$.
 Thus we see that $\varphi = d^* f + \varphi_0 \in \ker Q$. 
 Summarizing, we have $\varphi \in \ker Q \cap \ker (\varGamma -1)$.
 By Lemma \ref{lem:aqg}, we obtain $\varphi \in \ker \alpha$.  
 Thus (i) is proved.
 
 A similar proof works for (ii).
 \end{proof}
 
\subsection{Topological invariance}
\label{subsec:inv}
In this subsection, we prove Theorem \ref{thm:ind} (iii),
{\it i.e.}, if $(U, \varGamma)$ and $(U^\prime, \varGamma)$
are Fredholm pairs and $U^\prime - U$ is compact,
then ${\rm ind}_\varGamma(U^\prime) =  {\rm ind}_\varGamma(U)$.
\begin{proof}[Proof of Theorem \ref{thm:ind} (iii)]
Let $C_1 = \varGamma U$ and $C_2 = \varGamma U^\prime$.
By assumption,$C_i$ is written as $C_i = 2 P_i - 1$
with the projection onto $\ker(C_i -1)$ ($i=1,2$)
and  $2(P_1 - P_2) = C_1 - C_2 = \varGamma (U - U^\prime)$
is compact. 
Because supercharges $Q_i$ for $U_i:= \varGamma C_i$  ($i=1,2$)
are $Q_i = [\varGamma, P_i]/i$,
\begin{align*}
Q_1 - Q_2
& = \frac{1}{i} [\varGamma, P_1 - P_2] 
\end{align*} 
is compact. 
Let $\alpha_i = \varGamma_- Q_i \varGamma_+$. 
Because $\alpha_1 - \alpha_2 
= \varGamma_- (Q_1- Q_2) \varGamma_+$
is compact, the Fredholm index ${\rm index} (\alpha_1)$
is equal to ${\rm index} (\alpha_2)$. 
By definition, this means that 
${\rm ind}_\varGamma (U) = {\rm ind}_\varGamma (U^\prime)$. 
This completes the proof. 
\end{proof}
\begin{remark}
Similarly to the above proof,
we can relax the condition of Theorem \ref{thm:ind} (iii). 
Indeed, we can prove the following.
Suppose that $\varGamma U$ and $\varGamma U^\prime$ 
are unitary involutions and $U-U^\prime$ is compact.  
If $(U, \varGamma)$ is a Fredholm pair,
then so is  $(U^\prime, \varGamma)$ and 
${\rm ind}_\varGamma (U) = {\rm ind}_\varGamma (U^\prime)$. 
\end{remark}
\section{Examples}
\label{sec:ex}
Based on the supersymmetric structure discussed above, 
we will call a quantum walk with a chiral symmetric evolution 
a supersymmetric quantum walk (SUSYQW). 
After considering a finite dimensional toy model
in Subsection \ref{subsec:toy},
we present SUSYQWs.  
In Subsection \ref{subsec:GA}, 
we give an application to Grover's algorithm,
which is viewed as a SUSYQW. 
In Subsection  \ref{subsec:GW} 
we treat the Grover walk on a graph. 
\subsection{finite dimensional toy models}
\label{subsec:toy}
In this subsection,
we demonstrate how to calculate
the index ${\rm ind}_\varGamma (U)$ 
for a finite dimensional toy model,
which reveals that the index depends on the choice of $\varGamma$. 
For the finite dimensional case, 
Corollary \ref{cor:fin} says that 
the index ${\rm ind}_\varGamma (U)$ 
is given by the formula \eqref{eq:ind} for every pair $(U,\varGamma)$ 
of a unitary $U$ and a unitary involution $\varGamma$ obeying \eqref{GC}.
\begin{example}[Two dimensional case]
Fix $\beta \in \mathbb{R}$
and set 
\[ U = 
\begin{pmatrix}  e^{i \beta} & 0 \\ 
	0 & e^{-i \beta} \end{pmatrix}. \]
For $\gamma, c \in \mathbb{R}$
with $\beta = \gamma- c$,
it follows that $U = \varGamma C$, where 
\[ \varGamma 
= \begin{pmatrix} 0 & e^{i \gamma} \\ 
	e^{-i \gamma} & 0 \end{pmatrix},
\quad 
C
= \begin{pmatrix} 0 & e^{i c} \\ 
	e^{-i c} & 0 \end{pmatrix}.
\] 
Because $\varGamma$ and $C$ are unitary involutions,
$(U,\varGamma)$ becomes a Fredholm pair
for every $\gamma \in \mathbb{R}$
with $c = \gamma- \beta$.
Thus we find that there are
infinitely many choices of $\varGamma$ 
such that 	$(U,\varGamma)$ is a Fredholm pair. 
In this case, 
Corollary \ref{cor:fin} says that
${\rm ind}_\varGamma(U) = 0$,
because $\ker (\varGamma \pm 1) = 1$.

We next study 
the birth eigenspaces $\mathcal{B}_\pm$ 
and the inherited eigenspaces $\mathcal{T}_\pm$. 
Since ${\rm ind}_\varGamma(U) = 0$,
\[ M_+- M_- = m_+ - m_-. \]
Combining this with $\dim \ker(U\mp 1)= M_\pm +m_\pm$,
we can conclude the following assertion.
\begin{itemize} 
\item If $\beta \in \pi \mathbb{Z}$,
then either of the following two holds:\\
(i) $M_+ = m_+ = 1$, 
{\it i.e.}, $\mathcal{B}_- = \mathcal{T}_- = \{0\}$; \\
(ii) $M_- = m_- = 1$,
{\it i.e.}, $\mathcal{B}_+ = \mathcal{T}_+ = \{0\}$.  
\item Otherwise, $\ker (U - 1) = \ker (U+1) =\{0\}$ 
and hence
$M_+ = M_- = m_+ = m_- = 0$, 
{\it i.e.}, $\mathcal{B}_+ = \mathcal{B}_- 
= \mathcal{T}_+ = \mathcal{T}_- = \{0\}$. 
\end{itemize} 
The above assertions can be checked directly. 
Indeed, 
$\ker(\varGamma \mp 1) = {\rm span}\{ v_{\pm 1}(\gamma) \}$
and $\ker(C \mp 1) = {\rm span}\{ v_{\pm 1}(c) \}$,
where $v_{\pm 1}(\theta) = \begin{pmatrix} \pm e^{i \theta} \\ 1 \end{pmatrix}$ and 
$\langle v_j (\gamma), v_k(c) \rangle = 1+  j k e^{-i \beta}$
with $j, k = \pm 1$.  
For instance, in the case of $j=-1$, $k=+1$,
this implies  
$\mathcal{B}_+ = \ker (\varGamma + 1) \cap \ker (C - 1)
=\{0\}$ if $\beta \in 2 \pi \mathbb{Z}$. 
\end{example}
 \begin{example}[Four dimensional case] 
 \label{ex:4d}
 Let $\mathcal{H} = \mathbb{C}^4$
 and consider
 \[ U = \begin{pmatrix} 
 	1 & 0 & 0 & 0 \\
 	0 & 1 & 0 & 0 \\
 	0 & 0 & 1 & 0 \\
 	0 & 0 & 0 & -1
 	\end{pmatrix}. \] 
 The following table indicates
 $m_\pm$, $M_\pm$, and $I:={\rm ind}_\varGamma(U)$
 for several pairs $(\varGamma, C)$ of two unitary involutions 
 such that $U = \varGamma C$.
 \end{example}
 \begin{center}
\begin{tabular}{|c|c||c|c|c|c|c|}
\hline
$\varGamma$
& $C$
& $M_+$
& $M_-$  
& $m_+$
& $m_-$
& $I
$ \\ \hline \hline
$-1$
& 
$-U$
& $3$
& $0$
& $0$
& $1$
& $-4$
\\ \hline
$\begin{pmatrix} 
 	1 & 0 & 0 & 0 \\
 	0 & -1 & 0 & 0 \\
 	0 & 0 & -1 & 0 \\
 	0 & 0 & 0 & -1
 	\end{pmatrix}$
& $\begin{pmatrix} 
 	1 & 0 & 0 & 0 \\
 	0 & -1 & 0 & 0 \\
 	0 & 0 & -1 & 0 \\
 	0 & 0 & 0 & 1
 	\end{pmatrix}$
& $2$
& $0$  
& $1$
& $1$
& $-2$ \\ \hline 
$\begin{pmatrix} 
 	1 & 0 & 0 & 0 \\
 	0 & 1 & 0 & 0 \\
 	0 & 0 & -1 & 0 \\
 	0 & 0 & 0 & -1
 	\end{pmatrix}$
& $\begin{pmatrix} 
 	1 & 0 & 0 & 0 \\
 	0 & 1 & 0 & 0 \\
 	0 & 0 & -1 & 0 \\
 	0 & 0 & 0 & 1
 	\end{pmatrix}$
& $1$
& $0$  
& $2$
& $1$
& $0$ \\ \hline 
$\begin{pmatrix} 
 	1 & 0 & 0 & 0 \\
 	0 & -1 & 0 & 0 \\
 	0 & 0 & 1 & 0 \\
 	0 & 0 & 0 & 1
 	\end{pmatrix}$
& $\begin{pmatrix} 
 	1 & 0 & 0 & 0 \\
 	0 & -1 & 0 & 0 \\
 	0 & 0 & 1 & 0 \\
 	0 & 0 & 0 & -1
 	\end{pmatrix}$
& $1$
& $1$  
& $2$
& $0$
& $2$ \\ \hline 
$1$
& 
$U$
& $0$
& $1$  
& $3$
& $0$
& $4$ \\ \hline
\end{tabular}
\end{center}
\subsection{Grover's  search algorithm }
\label{subsec:GA}
Grover's searching algorithm \cite{Gr} 
consists of operators acting on the Hilbert space 
$\mathcal{H} 
= (\mathbb{C}^2)^{\otimes n} \otimes \mathbb{C}^2$,
where $(\mathbb{C}^2)^{\otimes n}$ 
describes $n$-qubit states 
and the oracle operator acts on $\mathbb{C}^2$. 
Let $N =2^n$ and $V = \{0, 1, \cdots, N-1\}$. 
We use $|x \rangle$ to denote
$|j_0 \rangle \otimes  \cdots \otimes |j_{n-1}\rangle 
\in (\mathbb{C}^2)^{\otimes n}$
where $\{|j \rangle \}_{j=0,1}$ 
is the standard basis of $\mathbb{C}^2$ and
 $j_i \in \{0,1 \}$ ($i=0,\cdots, n-1$) are the 2-adic digits,
{\it i.e.},  the 2-adic expansion of $x$
is given by $\sum_{i = 0}^{n-1} j_i 2^{i}$. 
It is useful to identify 
$(\mathbb{C}^2)^{\otimes n}$
with the Hilbert space $\ell^2(V)$ of functions on $V$,
in which case 
$|x\rangle$ is identified with a function $\delta_x$,
{\it i.e.}, $\delta_x(y) = 1$ if $y=x$ and $\delta_x(y) =0$ otherwise. 
With this identification,  
we write 
\[ \mathcal{H} = \ell^2(V) \otimes \mathbb{C}^2 \]
and consider the ONB $\{|x\rangle \otimes |\star \rangle \mid x \in V,  \star = \pm \}$ of $\mathcal{H}$,
where we use $|\pm \rangle$ to denote 
vectors $(|0\rangle \pm |1 \rangle)/\sqrt{2} \in \mathbb{C}^2$. 

We now introduce two operators on $\mathcal{H}$ known  as
the oracle operator and the diffusion operator. 
For a fixed $x_0 \in V$, 
we set $|\chi_0 \rangle = |x_0\rangle \otimes |-\rangle$.
The oracle operator is defined as
\[ C = 1 - 2|\chi_0 \rangle \langle \chi_0|. \]
Let $|\phi_0 \rangle = \sum_{x \in V}|x \rangle/\sqrt{N} \in \ell^2(V)$
and set $D_0 = 2|\phi_0 \rangle \langle \phi_0|-1$. 
The diffusion operator $\varGamma$ is
defined as 
\[ \varGamma = D_0 \otimes 1.  \]
Let $U = \varGamma C$. 
In Grover's algorithm,
after transforming the state $\Psi_0  = |\phi_0 \rangle \otimes |-\rangle$
by $U^t$,
we detect $x_0$ with a probability 
\[ p_t(x_0) = 
\|(|x\rangle \langle x|\otimes 1)U^t \Psi_0)\|_{\mathcal{H}}^2.
\] 
This is viewed as the probability of  finding  a quantum walker on $V$
at a position $x_0 \in V$. 
In this case, $U^t \Psi_0$ is the state of the walker at time $t$
when $\Psi_0$ is the initial state.   
From this viewpoint, 
$U = \varGamma C$ is the evolution operator of
a SUSYQW,
because $\varGamma$ and $C$ are unitary involutions,
as is easily verified.
 
In what follows, we calculate the index of $U$. 
\begin{theorem}
\label{thm:Gr}
Let $U$ and $\varGamma$ be stated as above. 
Then
\begin{equation}
\label{3:28} 
{\rm ind}_\varGamma(U) = 4-2N. 
\end{equation}	
Moreover
\begin{equation}
\label{3:30} 
\sigma(U) 
=\{e^{i \arccos (1 - 2/N)}, e^{-i \arccos (1 - 2/N)},1,-1\}.
\end{equation}
\end{theorem}
\begin{proof}
Form Corollary \ref{cor:-U},
it suffices to calculate the spectrum of 
$U^\prime:= -U = \varGamma C^\prime$,
where $C^\prime = -C = 2|\chi_0\rangle \langle \chi_0|-1$. 
To this end, we define an operator 
$d:\mathcal{H} \to \mathcal{K} := \mathbb{C}$ 
as 
\[ d \psi = \langle \chi_0,  \psi \rangle,
	\quad \psi \in \mathcal{H}. \]
Then the adjoint $d^*:\mathcal{K} \to \mathcal{H}$ is given by
\[ d^* f = f |\chi_0\rangle, \quad f \in \mathcal{K}. \]
It is straightforward to see that $d$ is a coisometry,
{\it i.e.}, $dd^* = 1$ on $\mathcal{K}$. 
The discriminant operator is then calculated as follows.
\begin{align*}
Tf & = d \varGamma d^*f 
	= \langle \chi_0, (D_0 \otimes 1) \chi_0 \rangle f \\
& = \left( 2 |\langle x_0, \phi_0 \rangle |^2 -1 \right) f \\
& = (2/N-1) f. 
\end{align*}
Hence, $T = 2/N -1 \not=0$ and $\sigma(T)=\{2/N-1\}$. 
In particular, because $\|T\| < 1$,
Corollary \ref{cor:T<1} says that 
\begin{equation} 
\label{18:46}
{\rm ind}_\varGamma(-U) = M_- - M_+ 
\end{equation}
with $M_\pm = \dim \ker (\varGamma \pm 1) \cap \ker d$
and $m_\pm = 0$. 

To count $M_\pm$, we calculate the spectrum of $-U$.  
By Theorem \ref{thm:smt},
\[ 
\sigma(-U) \setminus \{1, -1\}
= \varphi^{-1}(2/N-1) 
= \{e^{i \arccos (2/N-1)}, e^{-i \arccos (2/N-1)} \}
\]
and $\dim \ker (-U - e^{\pm i \arccos (2/N-1)}) = 1$.
Hence, 
\begin{equation}
\label{18:48} 
2N = \dim \mathcal{H} = M_+ + M_- + 2. 
\end{equation}
Observe that
\[ \mathcal{B}_- = \ker (\varGamma - 1) \cap \ker d
	= {\rm Ran}( |\phi_0 \rangle \langle \phi_0| \otimes 1 )
		 \cap {\rm Ran} (1- |\phi_0\rangle \otimes |- \rangle),
\]
where
\begin{align*}
{\rm Ran} \varGamma_+
	=  {\rm span}\{|\phi_0 \rangle \otimes |+\rangle, \
		|\phi_0 \rangle \otimes |-\rangle \} 
\end{align*}
and
\[ \ker d = 
	{\rm span}\{|x_0\rangle \otimes |+\rangle \}
		\oplus {\rm span}
			\{|x \rangle \otimes |\star\rangle 
				\mid \star = \pm, \ x \not=x_0 \}. \]
Hence,
\[ \mathcal{B}_- 
= {\rm span}\{|\phi_0 \rangle \otimes |+\rangle \} \]
and $M_- = 1$. Combining this with \eqref{18:46} and \eqref{18:48},
we obtain $M_+ = 2N-3$ and
\[ {\rm ind}_\varGamma(-U) = 1- (2N-3) = 4-2N. \]
Therefore, Corollary \ref{cor:-U} proves \eqref{3:28}. 

From the above argument, we observe that
\[ \sigma(U) = - \sigma(-U)
	 =\{- e^{i \arccos (2/N-1)}, -e^{-i \arccos (2/N-1)},1,-1\}.
\]
Because $- e^{\pm i \arccos (2/N-1)} 
= e^{\mp i \arccos (1 - 2/N)}$,
we obtain \eqref{3:30}. 
This completes the proof.
\end{proof}

%

\subsection{The Grover walk}
\label{subsec:GW}
Let $G = (V, E)$ be a connected undirected graph
(having multiple edges and self-loops)
with $V$ the set of vertices and $E$ the set of edges. 
For the undirected graph $G$, 
we introduce a set $D$ of directed edges of $G$ as follows. 
We first determine a direction for each edge $e \in E$
and denote the origin by $o(e)$ and the terminus by $t(e)$,
and next introduce  the inverse edge of $e$
by $o(\bar{e}) = t(e)$ and $t(\bar{e}) = o(e)$. 
We then define the set $D$ as all such directed edges.
By abuse of notation,
denoting the set of directed edges determined first 
by the same symbol $E$,
we can write $D = E \cup \bar{E}$,
where $\bar{E} = \{ \bar{e} \mid e \in E \}$. 
Following the definition in \cite{Se13, HKSS14}, 
we introduce the Grover walk on $G$ as follows . 
Let $\mathcal{H} = \ell^2(D)$ be the Hilbert space
of square summable functions on $D$.  
The shift operator $S$ is defined as
\[ (S\psi)(e) = \psi(\bar{e}),
	\quad e \in D, \ \psi \in \mathcal{H}. \] 
Let
\[ \chi_v 
	= \frac{1}{\sqrt{\deg v}}
		\sum_{e \in D: o(e)=v} \delta_e, 
\]
where $\deg v = \# \{ e \in D \mid o(e) = v\}$
and $\delta_e \in \mathcal{H}$ is defined by 
$\delta_{e}(f) = 1$ ($f=e$); $\delta_{e}(f)=0$ otherwise. 
Then a coisometry $d$ from $\mathcal{H}$ to 
$\mathcal{K}:= \ell^2(V)$ is defined as
\[ (d\psi)(v) = \langle \chi_v, \psi \rangle_{\mathcal{H}},
	\quad v \in V, \ \psi \in \mathcal{H}. \]
The coin operator $C$ is defined by
\[ C = 2d^*d-1. \]
Because $S$ is a unitary involution, $U$ is written as
\[ U = \varGamma C, \]
where $\varGamma = S$ and $C$ are unitary involutions. 
Hence the Grover walk is a SUSYQW. 
 
$M_\pm$ and $m_\pm$ have already been calculated in \cite{HKSS14} 
for finite graphs and several crystal lattices. 
See also \cite{HiSe1, HiSe2, HSS}
for magnifier graphs,  infinite trees, and the Sierpi\'nski lattice. 
It is noteworty that $M_\pm$ are determined by the number of cycles
and geometric properties of the graph.  
In particular, if the total number of all cycles  is infinity,
then $M_+ = \infty$.  
From \cite[Theorem 1 and Lemma 2]{HKSS14} 
and Theorem \ref{thm:ind}, we observe that
${\rm ind}_\varGamma(U) = 0$ for all finite graphs. 
For crystal lattices such as a triangular lattice,  
a square lattice, and a hexagonal lattice
$(U, \varGamma)$ are not Fredholm pairs,
because such graphs have infinitely many cycles.



\section*{Acknowledgments} \quad
This work was supported by JSPS KAKENHI Grant Number JP18K03327
and by the Research Institute for Mathematical Sciences, a Joint Usage/
Research Center located in Kyoto University.
The author thanks Y. Matsuzawa and Y. Tanaka
for helpful comments on the index formula for finite-dimensional cases 
and for one dimensional split-step quantum walks.  

\end{document}